\documentclass{article}

\usepackage{arxiv}
\usepackage{amsthm}
\usepackage{graphicx}
\usepackage{url}
\usepackage[utf8]{inputenc}
\usepackage{amsfonts}
\usepackage{amsmath}
\usepackage{amssymb, bm}
\usepackage{stmaryrd}
\usepackage{amssymb}
\usepackage{mathtools, bm}
\usepackage{algorithm}
\usepackage[noend]{algorithmic}
\usepackage{color}                      % colored text and backgrounds
\usepackage{hyperref}
\usepackage[toc,page]{appendix}

\newtheorem{problem}{Problem}
\newtheorem{theorem}{Theorem}
\newtheorem{lemma}[theorem]{Lemma}

\newcommand{\ZZ}{\mathbb{Z}}
\newcommand{\RR}{\mathbb{R}}
\newcommand{\eps}{\varepsilon}

\DeclareMathOperator{\conv}{conv}
\DeclareMathOperator{\area}{area}
\DeclareMathOperator{\polylog}{polylog}

\begin{document}

\title{Peeling Digital Potatoes}

\author{Loïc Crombez (1) \and
Guilherme D. da Fonseca (1) \and
Yan Gérard (1)
\\( (1) Universit\'e Clermont Auvergne and LIMOS,
 Clermont-Ferrand, France)}
%
% \authorrunning{L. Crombez et al.}
% First names are abbreviated in the running head.
% If there are more than two authors, 'et al.' is used.
%
% \institute{Universit\'e Clermont Auvergne and LIMOS,
%  Clermont-Ferrand, France}

\maketitle

\begin{abstract}
The potato-peeling problem (also known as convex skull) is a fundamental computational geometry problem that consist in finding the largest convex shape inside a given polygon. The fastest algorithm to date runs in $O(n^8)$ time for a polygon with $n$ vertices that may have holes. In this paper, we consider a digital version of the problem.
A set $K \subset \mathbb{Z}^2$ is \emph{digital convex} if $\conv(K) \cap \mathbb{Z}^2 = K$, where $\conv(K)$ denotes the convex hull of $K$. Given a set $S$ of $n$ lattice points, we present polynomial time algorithms for the problems of finding the largest digital convex subset $K$ of $S$ (\emph{digital potato-peeling problem}) and the largest union of two digital convex subsets of $S$. The two algorithms take roughly $O(n^3)$ and $O(n^9)$ time, respectively. We also show that those algorithms provide an approximation to the continuous versions.
\end{abstract}

%-----------------------------------------------------------------------
\section{Introduction}
%-----------------------------------------------------------------------

The \emph{potato-peeling problem}~\cite{Goo81} (also known as \emph{convex skull}~\cite{Woo86}) consists of finding the convex polygon of maximum area that is contained inside a given polygon (possibly with holes) with $n$ vertices. The fastest exact algorithm known takes $O(n^7)$ time without holes and $O(n^8)$ if there are holes~\cite{ChY86}.
The problem is arguably the simplest geometric problem for which the fastest exact algorithm known is a polynomial of high degree and this high complexity motivated the study of approximation algorithms~\cite{CCKS17,HKKM06}. Multiple variations of the problem have been considered, including triangle-mesh~\cite{AVLS11} and orthogonal~\cite{DBBB11,WoY88} versions.
In this paper, we consider a digital geometry version of the problem.

Digital geometry is the field of mathematics that studies the geometry of points with integer coordinates, also known as \emph{lattice points} \cite{KlR04}.
Different definitions of convexity in $\ZZ^2$ have been investigated, such as digital line, triangle, line~\cite{KR82}, HV (for Horizontal and Vertical~\cite{BDNP96}), and Q (for Quadrant~\cite{Da01}) convexities. These definitions guarantee that a digital convex set is connected (in terms of the induced grid subgraph), which simplifies several algorithmic problems.

Throughout this paper, however, we use the main and original definition of digital convexity from the geometry of numbers~\cite{Gru93}. A set of lattice points $K \subset \ZZ^d$ is \emph{digital convex} if $\conv(K) \cap \ZZ^d = K$, where $\conv(K)$ denotes the convex hull of $K$.
This definition does not guarantee connectivity of the grid subgraph, but provides several other important mathematical properties, such as being preserved under certain affine transformations. The authors recently showed how to  efficiently test digital convexity  in the plane~\cite{CDG19}. A natural question is to determine the largest digital convex subset.

The \emph{digital potato-peeling problem} is defined as follows and is illustrated in Figure~\ref{fig:pb_peeling}(a,b).

\begin{problem}[Digital potato-peeling] \label{prob:peeling}
Given a set $S \subset \ZZ^2$ of $n$ lattice points described by their coordinates, determine the \textit{largest} set $K \subseteq S$ that is digital convex (i.e., $\conv(K) \cap \ZZ^2 = K$), where largest refers to the area of $\conv(K)$.
\end{problem}

\begin{figure}[tb]
    \centering
        \includegraphics[scale=0.8]{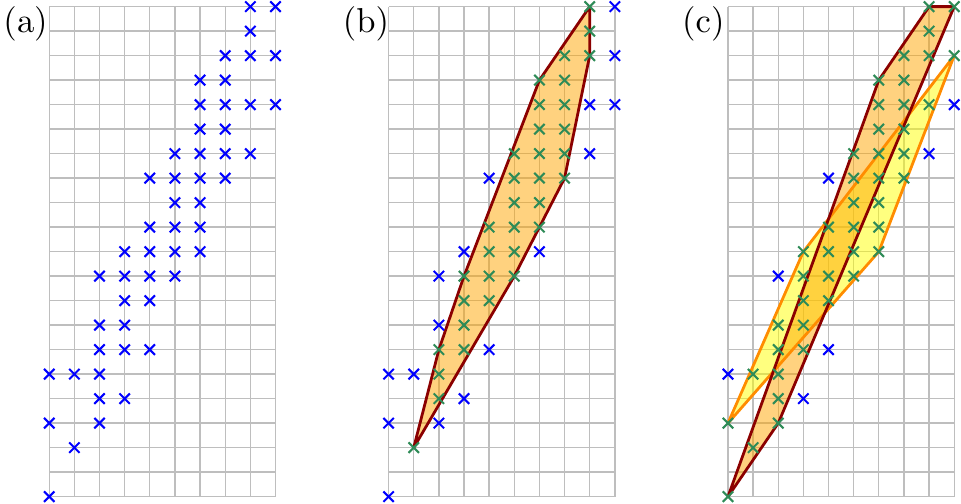}
    \caption{(a)~Input lattice set $S$. (b)~Largest digital convex subset of $S$ (Problem~\ref{prob:peeling}). (c)~Largest union of two digital convex subsets of $S$ (Problem~\ref{prob:2peeling}).}
    \label{fig:pb_2conv}
    \label{fig:pb_peeling}
\end{figure}

Our algorithms can easily be modified to maximize the number of points in $K$ instead of the area of $\conv(K)$. Compared to the continuous version, the digital geometry setting allows us to explicitly represent the whole set of input points, instead of limiting ourselves to polygonal shapes with polygonal holes. Note that the input of the continuous and digital problems is intrinsically different, hence we cannot compare the complexity of the two problems. 
Related continuous problems have been studied, such as the maximum volume of an empty convex body amidst $n$ points~\cite{DHT12}, or the \emph{optimal island problem}~\cite{BCD11,Fis97}, in which we are given two sets $S_p,S_n \subset \RR^2$, and the goal is to determine that largest subset $K \subseteq S_p$ such that $\conv(K) \cap S_n = \emptyset$.

Heuristics for the digital potato-peeling problem have been presented in~\cite{BoS05,ChC05}, but no exact algorithm was known. We solve this open problem by providing the first polynomial-time exact algorithm. 

We also solve the question of covering the largest area with two digital convex subsets. The problem is defined as follows and is illustrated in Figure~\ref{fig:pb_2conv}(a,c). 
\begin{problem}[Digital 2-potato peeling] \label{prob:2peeling}
Given a set $S \subset \ZZ^2$ of $n$ lattice points described by their coordinates, determine the largest set $K = K_1 \cup K_2 \subseteq S$ such that $K_1$ and $K_2$ are both digital convex, where largest refers to the area of $\conv(K_1) \cup \conv(K_2)$.
\end{problem}

A related continuous problem consists of completely covering a polygon by a small number of convex polygons inside of it. O'Rourke showed that covering a polygon with the minimum number of convex polygons is decidable~\cite{Oro82,Oro82-2}, but the problem has been shown to be NP-hard with or without holes~\cite{CuR94,Oro83}.
Shermer~\cite{She93} presents a linear time algorithm for the case of two convex polygons and Belleville~\cite{Bel93} provides a linear time algorithm for three. 
We are not aware of any previous results on finding a fixed (non-unit) number of convex polygons inside a given polygon and maximizing the area covered.

\subsection*{Our results}
We present polynomial time algorithms to solve each of these two problems. In Section~\ref{section:peeling}, we show how to solve the digital potato-peeling problem in $O(n^3 + n^2\log r)$ time, where $r$ is the diameter of the input $S$. 
We adapt an algorithm designed to solve the optimal island problem~\cite{BCD11,Fis97}. This algorithm builds the convex polygon $\conv(K)$ through its triangulation.
We use Pick's theorem~\cite{Pic1899} to test digital convexity for each triangle and the $O(\log r)$ factor in the running time comes from the gcd computation required to apply Pick's theorem.
The algorithm makes use of the following two properties: (i) it is possible to triangulate $K$ using only triangles that share a common bottom-most vertex $v$ and (ii) if the polygons lying on both sides of one such triangle (including the triangle itself) are convex, then the whole polygon is convex. 

These two properties are no longer valid for Problem~\ref{prob:2peeling}, in which the solution $\conv(K_1) \cup \conv(K_2)$ is the union of two convex polygons. Also, since convex shapes are not pseudo-disks (the boundaries may cross an arbitrarily large number of times), separating the input with a constant number of lines is not an option. Instead of property (i), our approach uses the fact that the union of two (intersecting) convex polygons can be triangulated with triangles that share a common vertex $\rho$ (that may not be a vertex of either convex polygon). Since $\rho$ may not have integer coordinates, we can no longer use Pick's theorem, and resort to the formulas from Beck and Robins~\cite{BeR02} or the algorithm from Barvinok~\cite{Bar94} to count the lattice points inside each triangle in $O(\polylog r)$ time.

Furthermore, to circumvent the fact that the solution no longer obeys property (ii), we use a directed acyclic graph (DAG) that encapsulates the orientation of the edges of both convex polygons. For those reasons, the running time of our algorithm for Problem~\ref{prob:2peeling} increases to $O(n^9 + n^6 \polylog r)$. The corresponding algorithm is described in Section~\ref{section:2peeling}.

In Section~\ref{section:continuous}, we show that a solution to the digital version of the problems provides an approximation to the continuous versions, establishing a formal connection between the continuous and digital versions.

Reducing the complexity of our algorithms or extending the result to higher numbers of convex polygons remain intriguing open questions, which are discussed in Section~\ref{section:conclusion}. Throughout, we assume the RAM model of computation, in which elementary operations on the input coordinates take constant time.

%-----------------------------------------------------------------------
\section{Digital Potato Peeling}\label{section:peeling}
%-----------------------------------------------------------------------

In this section, we present an algorithm to solve the digital potato-peeling problem in $O(n^3 + n^2 \log r)$ time, where $n$ is the number of input points and $r$ is the diameter of the point set. 

Fischer~\cite{Fis97} and Bautista et al.~\cite{BCD11} showed how to solve the following related problem in $O(n^3)$ time, where $n$ is the total number of points.
\begin{problem}[Optimal Island] \label{prob:optimal_island}
Given two sets $S_p,S_n \subset \RR^2$, determine the largest subset $K \subseteq S_p$ such that $\conv(K) \cap S_n = \emptyset$.
\end{problem}

The potato peeling problem~\ref{prob:peeling} for an input $S \subset \ZZ^2$ is the optimal island problem with $S_p = S$ and $S_n = \ZZ^2 \setminus S_p$. Restricting the problem to the bounding box of $S_p$, makes $S_n$ finite as $|S_n| = O(r^2)$. The resulting $O(r^6)$ complexity being very large relative to $r$, we do not use this direct approach. 
Nevertheless, the algorithm provides some key insights.

The algorithm consists of two phases. First, a list $\mathcal{T}$ of all \emph{valid} triangles is computed. A triangle $\triangle$ is said to be valid if its vertices are a subset of $S_p$ and if $\triangle \cap S_n = \emptyset$. Second, using $\mathcal{T}$ and the fact that every convex polygon has a fan triangulation in which all the triangles share a common bottom vertex, the solution is computed by appending valid triangles using dynamic programming.
In order to adapt this algorithm to solve the digital potato peeling, it suffices to compute the list of valid triangles $\mathcal{T}$.

\subsection{Valid Triangles}

For any triangle whose vertices are lattice points $\triangle$, and any digital set $S$: $|\triangle \cap S| = |\triangle \cap \ZZ^2|$ implies that $\triangle$ is valid.
As in~\cite{BCD11}, we use the following result of Eppstein et al.~\cite{EDO92} to compute $|\triangle \cap S|$.

\begin{theorem} \label{theorem:triangle_count}
Let $S$ be a set of $n$ points in the plane.
The set $S$ can be preprocessed in $O(n^2)$ time and space in order to, for
any query triangle $\triangle$ with vertices in $S$, compute the number of points $|\triangle \cap S|$ in constant time.
\end{theorem}

In order to compute $|\triangle \cap \ZZ^2|$, first, for all pairs of points $p_1,p_2 \in S$, we compute the number of lattice points lying on the edge $p_1p_2$ using a gcd computation. This takes $O(n^2 \log r)$ time, where $r$ is the diameter of $S$.
Now, using Pick's formula~\cite{Pic1899} which requires to compute both $area(\triangle)$ and the number of lattice points lying on the edges of $\triangle$, we determine in $O(1)$ time the validity of a triangle.  Since there are $O(n^3)$ triangles with vertices in $S$, the list $\mathcal{T}$ of all valid triangles is computed in $O(n^3 + n^2 \log r)$ time. Using $\mathcal{T}$, the algorithm of Bautista et al.~\cite{BCD11} determines the largest convex polygon formed by triangles in $\mathcal{T}$ in $O(n^3)$ time. Hence, we have the following theorem.

%-----------------------------------------------------------------------

\begin{theorem}
There exists an algorithm to solve Problem~\ref{prob:peeling} (digital potato peeling) in $O(n^3 + n^2 \log r)$ time, where $n$ is the number of input points and $r$ is the diameter of the input.
\end{theorem}

%-----------------------------------------------------------------------
\section{Digital 2-Potato Peeling} \label{section:2peeling}
%-----------------------------------------------------------------------

In this section, we show how to find two digital convex sets $K_1,K_2$, maximizing the area of $\conv(K_1) \cup \conv(K_2)$. We note that the solution described in this section can easily be adapted to solve the optimal 2-islands problem:
\begin{problem}[Optimal 2-Islands] \label{prob:optimal_2island}
Given two sets $S_p,S_n \subset \RR^2$, determine the largest union of subsets $K_1 \cup K_2$ such that $K_1 \cup K_2 \subseteq S_p$, $\conv(K_1) \cap S_n = \emptyset$ and $\conv(K_2) \cap S_n = \emptyset$.
\end{problem}

Consider a solution of the digital 2-potato peeling problem. Either the two convex hulls intersect or they do not (Figure~\ref{fig:opti_example}). We treat those two cases separately and the solution to Problem~\ref{prob:2peeling} is the largest among both. Hence, we consider the two following variations of the 2-potato-peeling problem.

\begin{problem}[Disjoint 2-potato peeling] \label{prob:separated_2}
Given a set $S \subset \ZZ^2$ of $n$ lattice points given by their coordinates, determine the \textit{largest} two digital convex sets $K_1 \cup K_2 \subseteq S$ such that $\conv(K_1) \cap \conv(K_2) = \emptyset$.
\end{problem}

\begin{problem}[Intersecting 2-potato peeling] \label{prob:intersecting_2}
Given a set $S \subset \ZZ^2$ of $n$ lattice points given by their coordinates, determine the \textit{largest} union of two digital convex sets $K_1 \cup K_2 \subseteq S$ such that $\conv(K_1) \cap \conv(K_2) \neq \emptyset$. In this case, largest means the maximum area of $\conv(K_1) \cup \conv(K_2)$.
\end{problem}

\begin{figure}[tb]
    \centering
\includegraphics[scale=0.8]{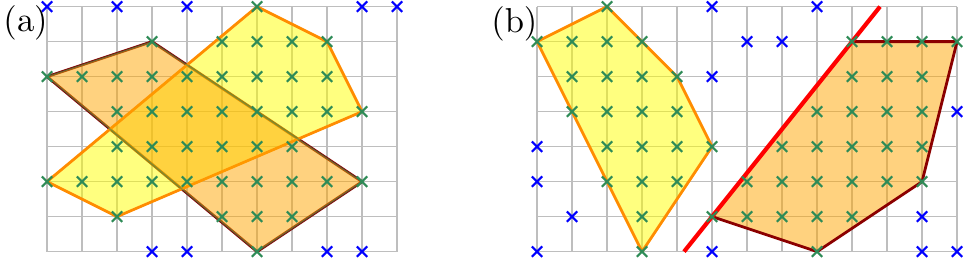}
    \caption{(a)~The two optimal sets intersect. (b)~The two optimal sets are disjoint and there is a supporting separating line.}
    \label{fig:opti_example}
\end{figure}

%-----------------------------------------------------------------------
\subsection{Disjoint Convex Polygons}
\label{section:no_intersect}

Any two disjoint convex shapes can be separated by a straight line. Moreover two convex polygons can be separated by a supporting line of an edge of one of the convex polygons (Figure~\ref{fig:opti_example}(b)).

For each ordered pair of distinct points $p_1,p_2 \in S$, we define two subsets $S_1,S_2$. The set $S_1$ contains the points on the line $p_1,p_2$ or to the 
left of it (according to the direction $p_2-p_1$). The set $S_2$ contains the remaining points.

For each pair of sets $S_1,S_2$, we independently solve Problem~\ref{prob:peeling} for each of $S_1$ and $S_2$. Since there are $O(n^2)$ pairs and each pair takes $O(n^3 + n^2 \log r)$ time, we solve Problem~\ref{prob:separated_2} in $O(n^5 + n^4\log r)$ time.

%-----------------------------------------------------------------------
\subsection{Intersecting Convex Polygons}
\label{section:intersect}

The more interesting case is when the two convex polygons intersect (Problem~\ref{prob:intersecting_2}). Note that it is possible to triangulate the union of two convex polygons that share a common boundary point $\rho$ using a fan triangulation around $\rho$ (Figure~\ref{fig:triangu_of_2_conv}). Hence we consider the following rooted version of the problem.

\begin{problem}[Rooted 2-potato peeling] \label{prob:2_rooted_peeling}
Given a set $S \subset \ZZ^2$ of $n$ lattice points represented by their coordinates and two edges $e_1,e_2 \in S^2$ that cross at a point $\rho$, determine the \textit{largest} union of two digital convex sets $K_1, K_2 \subseteq S$ such that $e_1$ is an edge of $\conv(K_1)$ and $e_2$ is an edge of $\conv(K_2)$.
\end{problem}

\begin{figure}[tb]
    \centering
    \includegraphics[scale=0.8]{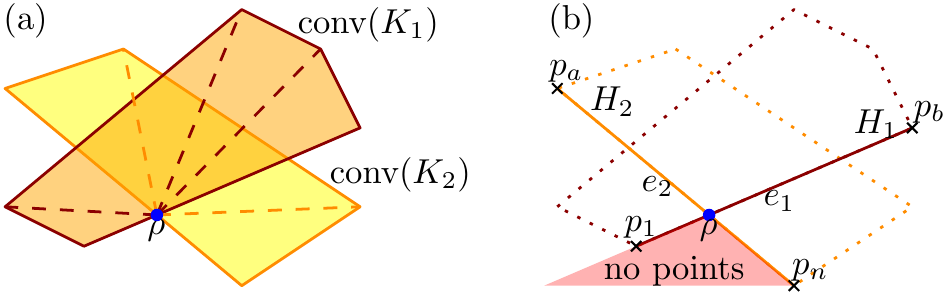}
    \caption{(a)~A fan triangulation of two intersecting convex polygons from a point $\rho$. (b)~Definitions used to solve Problem~\ref{prob:2_rooted_peeling}.}
    \label{fig:triangu_of_2_conv}
    \label{fig:point_denomination_2_peeling}
\end{figure}

Let $\rho$ be the intersection point of $e_1,e_2$.
The strategy of the algorithm to solve Problem~\ref{prob:2_rooted_peeling} is to encode the problem into a DAG $(V,E)$ whose longest directed path corresponds to the desired solution. To avoid confusion, we use the terms \emph{node} and \emph{arc} for the DAG and keep the terms \emph{vertex} and \emph{edge} for the polygons. It is well known that the longest directed path in a DAG $(V,E)$ can be calculated in $O(|V|+|E|)$ time~\cite{SeW11}.

Let $\mathcal{T}$ be the set of valid triangles with two vertices from $S$ and $\rho$ as the remaining vertex.
The nodes $V = \mathcal{T}^2 \cup \{v_0\}$ are ordered pairs of valid triangles and a starting node $v_0$.
% pairs of triangles $v=(\triangle_1, \triangle_2)$ from $\mathcal{T}$ (Fig.~\ref{fig:triangu_of_2_conv}), that is 
The number of nodes is $|V| = O(n^4)$. Before we define the arcs, we give an intuitive idea of our objective.

Each node $(\triangle_1,\triangle_2) \in V$ is such that $\triangle_1$ (resp. $\triangle_2$) is used to build the fan triangulation of $\conv(K_1)$ (resp. $\conv(K_2)$). The arcs will be defined in a way that, at each step as we walk through a path of the DAG, we add one triangle to either $\conv(K_1)$ or to $\conv(K_2)$. The arcs enforce the convexity of both $\conv(K_1)$ and $\conv(K_2)$. Furthermore, we enforce that we always append a triangle to the triangulation that is the least advanced of the two (in clockwise order), unless we have already reached the last triangle of $\conv(K_1)$. This last condition is important to allow us to define the arc lengths in a way that corresponds to the area of the union of the two convex polygons. Figure~\ref{fig:2peeling_example} illustrates the result of following a path on the DAG.

\begin{figure}[tb]
    \centering
    \includegraphics[scale=0.8]{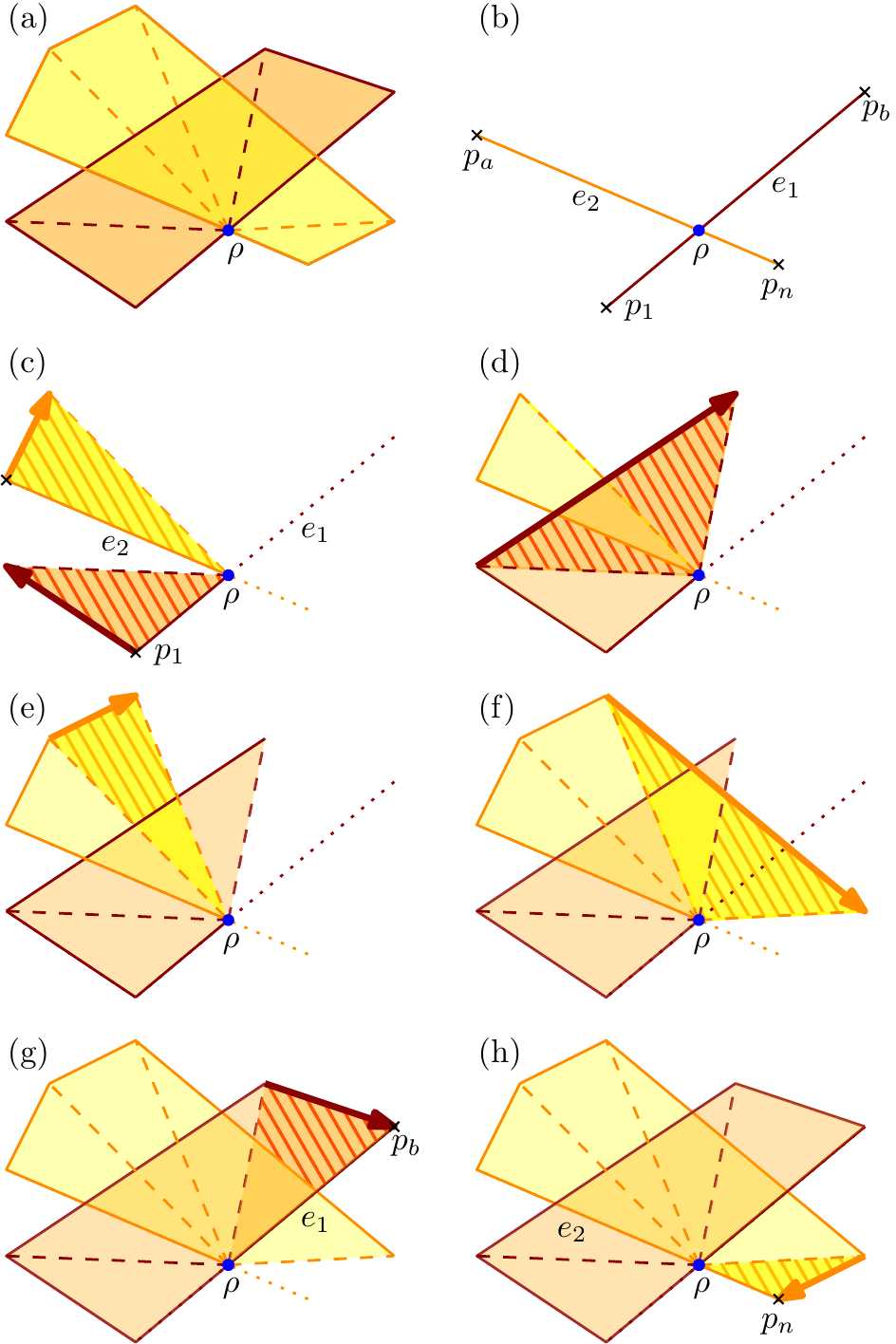}
    \caption{Steps of the algorithm from Section~\ref{section:intersect}. Figure~(a) represents the solution, while Figures~(b) to~(h) represent the triangulation obtained at each node of a path. The newly covered area that is assigned as the length of the corresponding arc is marked. In~(b), we have the initial pair of edges $e_1,e_2$ which corresponds to the starting vertex $v_0$. After following a type-0 arc, a first pair of triangles with vertices $p_1$ and $p_a$ is obtained in~(c). The triangle $\triangle _1$ is brown and triangle $\triangle _2$ yellow. From~(c) to~(d), we follow a type-1 arc. The triangle $\triangle _1$ (less advanced than triangle $\triangle _2$) advances. From~(d) to~(e), we follow a type-2 arc, since triangle $\triangle_2$ is less advanced. From~(e) to~(f) we have again a type-2 arc, and from~(f) to~(g) we have a type-1 arc. In~(g), the triangle $\triangle_1$ has reached the final node $p_b$ and cannot advance anymore. We have only type-2 arcs to follow until $\triangle_2$ reaches $p_n$, at a node in $V_1$.}  
    \label{fig:2peeling_example}
\end{figure}

The edge $e_1$ (respectively, $e_2$) from the problem input defines two halfplanes, one on each side. Let $H_1$ (resp. $H_2$) be the halfplane that contains $K_1$ (resp. $K_2$). We have not yet determined $K_1$ or $K_2$, but all four possibilities of halfplanes may be tried independently.
From now on, we only consider the $O(n)$ points of $S$ lying in the region $H_1 \cup H_2$.
Let $p_1,\ldots,p_n$ be the points of $S$ sorted clockwise around $\rho$, breaking ties arbitrarily. The edge $e_1$ (resp. $e_2$) has $p_1$ (resp. $p_n$) as a vertex. We define the indices $a < b$ such that $e_1 = (p_1, p_b), e_2 = (p_a,p_n)$ (Figure~\ref{fig:point_denomination_2_peeling}).

We are now ready to define the set $E$ of arcs of the DAG. There are three types of arcs. 
The \emph{type-0} arcs start from the initial node $v_0$ to $(\triangle_1,\triangle_2)$ if $p_1$ is a vertex of $\triangle_1$ and $p_a$ a vertex of $\triangle_2$. These two triangles of vertices $\rho,p_1,p_j$ with $j>1$  and $\rho,p_a,p_v$ with $v>a$ are respectively bounded by the edges $e_1$ and $e_2$. They initialize the triangulations of our two polygons $\conv(K_1)$ and $\conv(K_2)$. There are $O(n^2)$ type-0 arcs.

A \emph{type-1} arc corresponds to advancing the triangulation of $\conv(K_1)$, while a \emph{type-2} arc corresponds to advancing the triangulation of $\conv(K_2)$. There are $O(n)$ type-$1,2$ arcs coming out of each node. A \emph{type-1} arc goes from $(\triangle_1,\triangle_2)$ to $(\triangle_3,\triangle_2)$ if:

\begin{itemize}
    \item the quadrilateral $\triangle_1 \cup \triangle_3$ is convex,
    \item $\triangle_1$ has vertices $\rho,p_i,p_j$ with $i<j<b$,
    \item $\triangle_2$ has vertices $\rho,p_u,p_v$ with $a \leq u < v$,
    \item $\triangle_3$ has vertices $\rho,p_j,p_k$ with $j<k \leq b$,
    \item and $j \leq v$.
\end{itemize}

Similarly, there is a type-2 arc from $(\triangle_1,\triangle_2)$ to $(\triangle_1,\triangle_4)$ if:

\begin{itemize}
    \item the quadrilateral $\triangle_2 \cup \triangle_4$ is convex,
    \item $\triangle_1$ has vertices $\rho,p_i,p_j$ with $i<j \leq b$,
    \item $\triangle_2$ has vertices $\rho,p_u,p_v$ with $a \leq u < v$,
    \item $\triangle_4$ has vertices $\rho,p_v,p_w$ with $v<w$,
    \item and either $v \leq j$ or $j = b$.
\end{itemize}

The length of each arc corresponds to the area of the new region covered by appending a new triangle by following the arc.
Therefore, the length of a type-0  arc from $v_0$ to $(\triangle_1,\triangle_2)$ is the area of $\triangle_1 \cup \triangle_2$. 
The length of a type-1  arc from $(\triangle_1,\triangle_2)$ to $(\triangle_3,\triangle_2)$ is defined as the area of $\triangle_3 \setminus \triangle_2$. Similarly, the length of a type-2  arc from $(\triangle_1,\triangle_2)$ to $(\triangle_1,\triangle_4)$ is defined as the area of $\triangle_4 \setminus \triangle_1$.

We define a set of \emph{end} nodes $V_1$ as follows. A node $(\triangle_1,\triangle_2)$ is an end node if $p_b$ is a vertex of $\triangle_1$ and $p_n$ is a vertex of $\triangle_2$.
The construction of the DAG allows us to prove the following lemma.

\begin{lemma}\label{T1}
There is a bijection between the directed paths of the DAG $(V,E)$ (starting from $v_0$ and ending in $V_1$) and the digital convex sets $K_1,K_2 \subset S$ such that $e_1$ is an edge of $\conv(K_1)$ and $e_2$ is an edge of $\conv(K_2)$. Furthermore, the length of each path is equal to the corresponding area of $\conv(K_1) \cup \conv(K_2)$. (We assume that $K_1$ (resp. $K_2$) lie above the supporting line of $e_1$ (resp. $e_2$).)
\end{lemma}

\begin{proof}
First we show that the existence of two digital convex sets $K_1,K_2 \subset S$ as in the lemma statement implies the existence of a directed path in the DAG as in the lemma statement.
Let $K_1$ (resp. $K_2$) be two convex sets lying above the supporting line of $e_1$ (resp. $e_2$). Both $\conv(K_1)$ and $\conv(K_2)$ contain $\rho$ as a boundary point and hence can be triangulated from $\rho$. It is easy to see that there is a path corresponding to this triangulation. Next, we show that the converse also holds.

The definition of the arcs is such that advancing through one of them adds a triangle to one of the two polygons while preserving convexity, which ensures that all paths correspond to convex polygons. Furthermore, the starting node ensures that the two convex polygons respectively start from $p_1$ and $p_a$, while the set of ending nodes ensure that the two convex polygons respectively end at $p_b$ and $p_n$. Hence all paths from $v_0$ to $V_1$ correspond to two convex polygons that fit the lemma statement, one from edge $e_1 = p_1,p_b$  and one from edge $e_2 = p_a,p_n$.
The validity test on each triangle ensures that the paths describes digital convex sets.

The definition of the arcs enforces that we only move forward the least advanced triangle, that is the triangle that has the minimum maximum index among its vertices. The only exception is when $\conv(K_1)$ is completed, that is the triangle with vertex $p_b$ has been added to its triangulation. This ensures that the new area covered by a type-$1,2$ arc is simply the set theoretic difference of two triangles (instead of a triangle and an arbitrary convex object).
As the length of the arcs is defined as the area of the difference of the two triangles, the total length of the path is equal to the area of the union of the two convex polygons. Hence each path from $v_0$ to $V_1$ describe two digital convex sets $K_1, K_2 \in S$ such that $e_1$ is an edge of $\conv(K_1)$ and $e_2$ is an edge of $\conv(K_2)$, and the length of each path is equal to the corresponding area of $\conv(K_1) \cup \conv(K_2)$. 
\end{proof}

\begin{theorem}
There exists an algorithm to solve Problem~\ref{prob:2peeling} (digital 2-potato peeling) in $O(n^9 + n^6 \polylog r)$ time, where $n$ is the number of input points and $r$ is the diameter of the input.
\end{theorem}

\begin{proof}
As explained in Section~\ref{section:no_intersect}, solving the disjoint case (Problem~\ref{prob:separated_2}) takes $O(n^5 + n^4 \log r)$ time. Next, we show how to solve the rooted intersecting case (Problem~\ref{prob:2_rooted_peeling}) in $O(n^5 + n^2 \polylog r)$ time, proving the theorem.

Assume without loss of generality that $K_1,K_2$ are respectively above the supporting lines of $e_1,e_2$ (all four possibilities may be tried independently).

Our algorithm starts by computing the DAG $(V,E)$ with $O(n^4)$ nodes, each representing a pair of triangles. Since each node has at most $O(n)$ incoming arcs, the number of arcs is $O(n^5)$. Hence the longest path can be found in $O(n^5)$ time.

To build the set of nodes $V$, we need to test the validity of $O(n^2)$ triangles. Since $\rho$ may not be a lattice point, Pick's theorem~\cite{Pic1899} cannot be used. Still, $\rho$ is a rational point with denominators bounded by $O(r^2)$. Hence, we can use either the formulas from Beck and Robins~\cite{BeR02} or the algorithm from Barvinok~\cite{Bar94} to calculate the number of lattice points $|T \cap \ZZ^2|$ inside each triangle $T$ in $O(\polylog r)$ time.
As in Section~\ref{section:peeling}, we compute $|T \cap S|$ using a triangle range counting query, which takes $O(\log n)$ time after preprocessing $S$ in $O(n^2)$ time~\cite{CSW92}.
The triangle is valid if and only if $|T \cap \ZZ^2| = |T \cap S|$.
The two steps to test the validity of a triangle take $O(\polylog r)$ and $O(\log n)$ time. Since the diameter $r$ of $n$ lattice points is $\Omega(\sqrt{n})$, the dominating term is $O(\polylog r)$.
Hence, we test the validity of each triangle in $O(\polylog r)$ time, which gives a total time of $O(n^2 \polylog r)$ to build the list of valid triangles required to build $V$.

Consequently, we solve Problem~\ref{prob:2_rooted_peeling} in $O(n^5 + n^2 \polylog r)$ time. To obtain a solution to Problem~\ref{prob:2peeling}, we note that there are $O(n^2)$ candidates for the edge $e_1$, as well as for the edge $e_2$. Testing all $O(n^4)$ possible edges $e_1,e_2$, we achieve the claimed running time of $O(n^9 + n^6 \polylog r)$ time. 
\end{proof}

%-----------------------------------------------------------------------
\section{From Digital to Continuous} \label{section:continuous}
%-----------------------------------------------------------------------

In this section, we show that the exact algorithms for the digital potato-peeling problem and the digital 2-potato-peeling problem can be used to compute an approximation of the respective continuous problems with an arbitrarily small approximation error. For simplicity, we focus on the potato-peeling problem, but the 2-potato-peeling case is analogous.
We note that the reduction presented here does not lead to efficient approximation algorithms and is presented only to formally connect the continuous and digital versions of the problem.

\begin{problem}[Continuous potato-peeling] \label{prob:peeling_cont}
Given a polygon $P$ (that may have holes) of $n$ vertices, determine the \textit{largest} convex polygon $K \subseteq P$, where largest refers to the area of $K$.
\end{problem}

We start with some definitions. Let $K_C$ be the polygon of the optimal solution to the continuous problem above and $A_C$ be the area of $K_C$. 
Given an approximation parameter $\eps > 0$, we show how to obtain a set of lattice points $S \subseteq P$ such that the area $A_D$ of the convex hull of the solution $K_D$ of Problem~\ref{prob:peeling} with input $S$ satisfies $|A_C - A_D| = O(r \eps)$. In this section, we use lattice points that are not integers, but points with coordinates that are multiples of $\eps$. Let $\Lambda_\eps$ denote the set of all points with coordinates that are multiple of $\eps$. Of course, a uniform scaling maps $\Lambda_\eps$  to the integer lattice used in the remainder of the paper, and hence the integer lattice algorithms also apply to $\Lambda_\eps$. 

For a polygon $P$, the \emph{erosion} of $P$, denoted $P^-$ is the subset of $P$ formed by points within $L_\infty$ distance at least $2 \eps$ of all points outside $P$ (Figure~\ref{fig:p_and_pminus}(a)). Let $A^-_C$ be the area of the optimal solution to the continuous potato-peeling problem with input $P^-$.

We only give here the main directions of the proof. A more detailed version of the proof can be found in the appendix~\ref{appendix}.
First, by bounding the number of lattice cells that a convex curve of a given length can cross, we bound by $O(r\eps)$ the area difference between any convex polygon and the convex hull of the lattice points inside it.
Then, we use an erosion of $2\eps$ in order to smooth the input and avoid difficulties related to comb like input polygons. We bound the area difference by $O(r\eps)$ between the solution of problem~\ref{prob:peeling} for any polygon and the solution for the erosion of this polygon.
Finally, despite the digital solution being potentially outside the input polygon, it can be shown that the area lying outside the input polygon is bounded by $O(r\eps)$ which gives us the following theorem:

\begin{theorem}
Let $A_C$ be the area of the solution $K_C$ of Problem~\ref{prob:peeling_cont} with input polygon $P$ of diameter $r$. Let $\eps>0$ be a parameter and $S = \Lambda_\eps \cap P^-$, where $P^-$ is the erosion of $P$ by $2\eps$ and $\Lambda_\eps$ is the lattice of size $\eps$. The area $A_D$ of the convex hull of the solution $K_D$ of Problem~\ref{prob:peeling} with input $S$ satisfies
$|A_C - A_D| = O(r \eps)$.
\end{theorem}

The polygon $\conv(K_D)$ in the previous theorem may partially extend outside $P$. Nevertheless, the solution $K_D$ of Problem~\ref{prob:peeling} can be used to obtain a convex polygon $K \subseteq P$ which has an area $A$ satisfying $|A_C - A_D| = O(r \eps)$.

%-----------------------------------------------------------------------
\section{Conclusion and Open Problems}
\label{section:conclusion}
%-----------------------------------------------------------------------

The (continuous) potato peeling problem is a very peculiar problem in computational geometry. The fastest algorithms known have running times that are polynomials of substantially high degree. Also, we are not aware of any algorithms (or difficulty results) for the natural extensions to higher dimensions (even 3d) or to a fixed number of convex bodies.

In this paper, we focused on a digital version of the problem. Many problems in the intersection of digital, convex, and computational geometry remain open. Our study falls in the following framework of problems, all of which receive as input a set of $n$ lattice points $S\subset \ZZ^d$ for constant $d$ and are based on a fixed parameter $k \geq 1$.

\begin{enumerate}
    \item Is $S$ the union of at most $k$ digital convex sets?
    \item What is the smallest superset of $S$ that is the union of at most $k$ digital convex sets?
    \item What is the largest subset of $S$ that is the union of at most $k$ digital convex sets?
\end{enumerate}

In~\cite{CDG19}, the authors considered the first problem for $k=1$, presenting polynomial time solutions (which may still leave room for major improvements for $d>3$). We are not aware of any previous solutions for $k>1$. In contrast, the continuous version of the problem is well studied. The case of $k=1$ can be solved easily by a convex hull computation or by linear programming. Polynomial algorithms are known for $d=2$ and $k\leq3$~\cite{Bel93,She93}, as well as for $d=3$ and $k \leq 2$~\cite{Bel95b}. The problem is already NP-complete for $d=k=3$~\cite{Bel95b}. Hence, the continuous version remains open only for $d=2$ and fixed $k > 3$.

It is easy to obtain polynomial time algorithms for the second problem when $k=1$, since the solution consists of all points in the convex hull of $S$. The continuous version for $d=k=2$ can be solved in $O(n^4 \log n)$ time~\cite{BCES18}. Also, the orthogonal version of the problem is well studied (see for example~\cite{EGH16}). We know of no results for the digital version.

In this paper, we considered the digital version of the third problem for $d=2$ and $k = 1,2$, presenting algorithms with respective running times of $O(n^3 + n^2 \log r)$ and $O(n^9 + n^6 \polylog r)$, where $r$ is the diameter of $S$. Since the first problem trivially reduces to the third problem, we also solved the first problem for $k=d=2$ in $O(n^9 + n^6 \polylog r)$ time. It is surprising that we are not aware of any faster algorithm for the first problem in this particular case. 

The third problem for $d > 2$ or $k>2$ remains open. The DAG approach that we used for $d=2$ is unlikely to generalize to higher dimensions, since there is no longer a single order by which to transverse the boundary of a convex polytope. Surprisingly, even the continuous version seems to be unresolved for $d > 2$ or $k \geq 2$.

%-----------------------------------------------------------------------
% Bibliography
%-----------------------------------------------------------------------
\bibliographystyle{plainurl}
\bibliography{sample}

\newpage
\section*{Appendix}\label{appendix}
%-----------------------------------------------------------------------
\section{From Digital to Continuous}
%-----------------------------------------------------------------------

The \emph{width} of $P$ is the minimum distance between two parallel lines $\ell_1,\ell_2$ such that $P$ is between $\ell_1$ and $\ell_2$.

The following lemma that bounds the area difference between a convex polygon and the convex hull of its intersection with a lattice set will be useful to our proof.

\begin{lemma}\label{lemma:dig_to_cont}
Let $C$ be a convex polygon of diameter $r$. The convex hull $H = \conv(C \cap \Lambda_\eps)$ satisfies
\[\area(C) \leq \area(H) + 6 \sqrt{2} \pi r \eps + 16\eps^2.\]
\end{lemma}
\begin{proof}
The lattice $\Lambda_\eps$ induces a grid with vertex set $\Lambda_\eps$ and square cells of side length $\eps$. Let $X^-$ be the set of grid cells that are completely contained in $C$ and $X^\partial$ be the set of cells that are partially contained in $C$. All cells in $X^\partial$ intersect the boundary $\partial C$ of $C$.

Since the perimeter of a convex shape is at most $\pi$ times its diameter~\cite{AHM07}, the perimeter of $\partial C$ is at most $\pi r$. Since a curve of perimeter $p$ intersects at most $3 p / \eps \sqrt{2}  + 4$ grid cells of side length $\eps$ \cite{GeVF16}, we have $|X^\partial| \leq 3 \pi r / \eps \sqrt{2} + 4$.

All cells in $X^-$ are contained in $H$ and $C$ is covered by $X^- \cup X^\partial$. Therefore, the area of $C\setminus H$ is at most the area in $X^\partial$, which is 
\[\eps^2 |X^\partial| \leq 4\eps^2 \cdot \left(\frac{3}{\eps\sqrt{2}} \pi r + 4\right) = 6 \sqrt{2} \pi r \eps + 16\eps^2,\]
proving the lemma. 
\end{proof}

The following lemma bounds the area difference between the optimal solutions of the continuous potato peeling problem with inputs $P$ and $P^-$.

\begin{lemma}\label{lemma:dig_to_cont2}
Let $P$ be a polygon of diameter $r$ and $P^-$ be the erosion of $P$. Let $C$ (resp. $C'$) denote the largest convex polygon inside $P$ (resp. $P^-$). We have the following inequality:
\[\area(C) \leq \area(C') + 2 \sqrt{2} \pi r \eps.\]
\end{lemma}

\begin{proof}
    The erosion $C^-$ of $C$ is a convex polygon that lies inside $P^-$. Hence the area of $C'$ is at least as large as the area of $C^-$.

     As $C$ is a convex polygon of diameter at most $r$, the perimeter of $C$ is at most $\pi r$. As every eroded points from $C$ in order to obtain $C^-$ are inside $C$ and at a maximum distance of $2 \sqrt{2} \eps$ of the boundary of $C$, they are all included inside a set of rectangles that lie inside $C$ with the edges of $C$ as sides and width $2 \sqrt{2} \eps$. 
    Hence, the area difference between $C$ and its erosion is at most $2 \sqrt{2} \eps \pi r$, which proves the lemma.
\end{proof}

\begin{figure}[tb]
    \centering
    \includegraphics[scale=0.8]{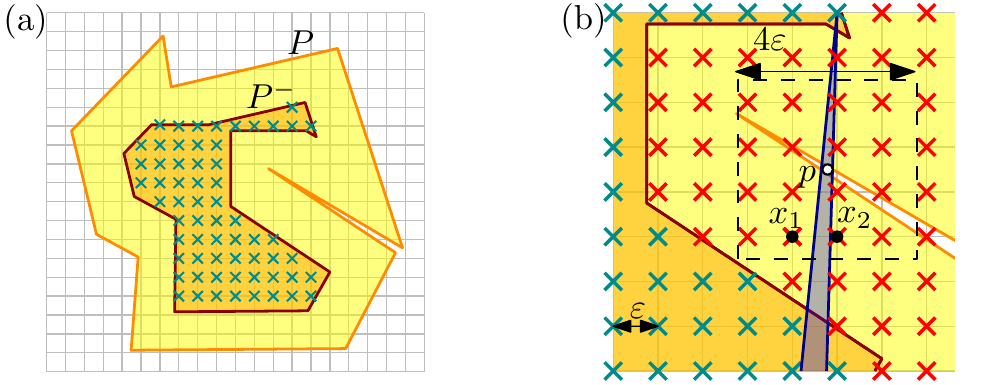}
    \caption{(a)~A polygon $P$, its erosion $P^-$, and the set $\Lambda_\eps \cap P^-$. (b)~To include the point $p$ that is outside $P$, $\conv(K_D)$ has to go between the lattice points within $L_\infty$ distance $2 \eps$ of $p$. }
    \label{fig:p_and_pminus}
\end{figure}

The digital solution may have portions that lie outside the input polygon $P$ of the continuous version. However, this portion cannot be too big, as shown in the following lemma.

\begin{lemma}\label{lemma:dig_to_cont3}
Let $P$ be a polygon of diameter $r$ and $P^-$ be the erosion of $P$. Let $S = P^- \cap \Lambda_\eps$, and $K_D$ be the largest digital convex subset of $S$. The following inequality holds:
\[\area(\conv(K_D) \setminus P) \leq 2 r \eps.\]
\end{lemma}
\begin{proof}
    Let $p$ be a point in $\conv(K_D) \setminus P$. 
    As $S$ is included inside $P^-$, all the lattice points within $L_\infty$ distance $\eps$ of $p$ are not in $S$ (see Figure~\ref{fig:p_and_pminus}).
    All $16$ lattice points at a $L_\infty$ distance less than $2 \eps$ of $p$ are not in $S$. Hence, in order to include $p$, $\conv(K_D)$ has to lie between two vertically (or horizontally) consecutive lattice points $x_1$ and $x_2$, which are separated by distance $\eps$. Furthermore $p$ is at a horizontal (or vertical) distance strictly greater than $\eps$ from $x_1$ and $x_2$.
    The widest angle the incoming and outgoing edges of $C$ can form is hence $2\arctan(1/2)$, effectively forming a turning angle of at least $\pi - 2\arctan(1/2)$.
    As the sum of turning angles inside a convex polygon is equal to $2 \pi$ and can never decrease, and as $\pi - 2\arctan(1/2) > 2\pi/3$ such a turning angle can only happen twice. 
    Also, as in order to include any point $p$ outside of $P$, $\conv(K_D)$ has to go in between $x_1$ and $x_2$, the width of this (possible non-contiguous region) including $p$ is at most $\eps$ and the diameter at most $r$, hence, the area is bounded by $r\eps$.
    Therefore, there can be no more than two such regions in $\conv(K_D)$ (even though each of them can enter and leave $P$ multiple times), which proves the lemma.
\end{proof}

% Thanks to lemma~\ref{lemma:dig_to_cont} and lemma~\ref{lemma:dig_to_cont2}, we can conclude that:
Using lemma~\ref{lemma:dig_to_cont2}, it follows that $A_C - A^-_C \leq 2 \sqrt{2} \pi r \eps$.
Lemma~\ref{lemma:dig_to_cont} gives us that $A^-_C - A_D \leq 6 \sqrt{2} \pi r \eps + 16\eps^2$.
Lemma~\ref{lemma:dig_to_cont3} gives us that %$A_D - A_C \leq 2 r \eps$  
$A_D - 2 r \eps \leq A_C$. Hence,
\[A_C - 8 \sqrt{2} \pi r \eps-16\eps^2 \leq A_D \leq A_C + 2r\eps,\]
proving the following theorem.

\begin{theorem}
Let $A_C$ be the area of the solution $K_C$ of Problem~\ref{prob:peeling_cont} with input polygon $P$ of diameter $r$. Let $\eps>0$ be a parameter and $S = \Lambda_\eps \cap P^-$, where $P^-$ is the erosion of $P$ by $2\eps$ and $\Lambda_\eps$ is the lattice of size $\eps$. The area $A_D$ of the convex hull of the solution $K_D$ of Problem~\ref{prob:peeling} with input $S$ satisfies
$|A_C - A_D| = O(r \eps)$.
\end{theorem}

The polygon $\conv(K_D)$ in the previous theorem may partially extend outside $P$. Nevertheless, the solution $K_D$ of Problem~\ref{prob:peeling} can be used to obtain a convex polygon $K \subseteq P$ which has an area $A$ satisfying $|A_C - A_D| = O(r \eps)$.

The same proof strategy can be applied to obtain an approximation to the continuous version of the 2-potato-peeling problem using the digital version of the problem.
\end{document}